\documentclass[journal,comsoc]{IEEEtran}
\usepackage[T1]{fontenc}% optional T1 font encoding

\usepackage{subfigure}
\usepackage{tabularx}
\usepackage{xcolor}
\usepackage{multirow}
\usepackage{booktabs}
\usepackage{array}
\newcolumntype{L}[1]{>{\raggedright\let\newline\\\arraybackslash\hspace{0pt}}m{#1}}
\newcolumntype{C}[1]{>{\centering\let\newline\\\arraybackslash\hspace{0pt}}m{#1}}
\newcolumntype{R}[1]{>{\raggedleft\let\newline\\\arraybackslash\hspace{0pt}}m{#1}}

\usepackage[english]{babel}
\usepackage{amsthm}
\usepackage{amssymb}

\newtheorem{theorem}{Theorem}
\newtheorem{lemma}{Lemma}

\theoremstyle{plain}
\newtheorem{corollary}{Corollary}

\theoremstyle{plain}

\theoremstyle{remark}

\usepackage{cite}
\ifCLASSINFOpdf
   \usepackage[pdftex]{graphicx}
\else
   or other class option (dvipsone, dvipdf, if not using dvips). graphicx
   \usepackage[dvips]{graphicx}
   \graphicspath{{../eps/}}
   \DeclareGraphicsExtensions{.eps}
\fi
\usepackage[normalem]{ulem}

\usepackage{amsmath}
\usepackage[cmintegrals]{newtxmath}
\usepackage{algorithmic}
\usepackage{multicol}
\usepackage{algorithm}

\usepackage{verbatim}

\usepackage{array}
\usepackage{makecell}
\usepackage[pagewise]{lineno}
\usepackage{upgreek}

\ifCLASSOPTIONcompsoc
  \usepackage[caption=false,font=normalsize,uabelfont=sf,textfont=sf]{subfig}
\else
  \usepackage[caption=false,font=footnotesize]{subfig}
\fi
\begin{document}

% \linenumbers
\setcounter{figure}{0}
\renewcommand{\figurename}{Fig.}
\renewcommand{\thefigure}{\arabic{figure}}
% paper title
% Titles are generally capitalized except for words such as a, an, and, as,
% at, but, by, for, in, nor, of, on, or, the, to and up, which are usually
% not capitalized unless they are the first or last word of the title.
% Linebreaks \\ can be used within to get better formatting as desired.
% Do not put math or special symbols in the title.
% \title{Sensing with OFDM Communication Signals: Performance Analysis, Constellation Shaping, and Experimental Validation}
\title{OFDM-ISAC over Data Payloads: MSE Analysis, Constellation Design, and Experimentation}
%\pagestyle{empty}
%
% author names and IEEE memberships
% note positions of commas and nonbreaking spaces ( ~ ) LaTeX will not break
% a structure at a ~ so this keeps an author's name from being broken across
% two lines.
% use \thanks{} to gain access to the first footnote area
% a separate \thanks must be used for each paragraph as LaTeX2e's \thanks
% was not built to handle multiple paragraphs
%

\author{Kawon~Han,~\IEEEmembership{Member,~IEEE,} Kaitao~Meng,~\IEEEmembership{Member,~IEEE,}  Alexandra Chatzicharistou,~\IEEEmembership{Graduate Student Member,~IEEE,} and Christos~Masouros,~\IEEEmembership{Fellow,~IEEE}

\thanks{An earlier version of this paper is accepted in part at the 2026 IEEE International Conference on Communications \cite{han2025constellation}

K. Han is with the Department of Electrical Engineering, Ulsan National Institute of Science and Technology (UNIST), Ulsan, South Korea (emails:kawon.han@unist.ac.kr). 

K. Meng is with the Department of Electrical and Electronic Engineering, University of Manchester, Manchester, UK (emails:kaitao.meng@manchester.ac.uk). 

A. Chatzicharistou and C. Masouros are with the Department of Electronic and Electrical Engineering, University College London, London, UK (emails: \{a.chatzicharistou, c.masouros\}@ucl.ac.uk).}}

\maketitle

% As a general rule, do not put math, special symbols or citations
% in the abstract or keywords.
\begin{abstract}
Orthogonal frequency division multiplexing (OFDM) is a key waveform for integrated sensing and communication (ISAC) systems due to its high spectral efficiency and inherent compatibility with modern wireless standards. However, its fundamental estimation-theoretic sensing performance under random data modulation remains largely unexplored. This paper presents a unified and explicit performance analysis of OFDM-based ISAC systems for multi-target range estimation, focusing on the distinct impacts of the modulation constellation on the sensing performance. We develop a comprehensive estimation-theoretic framework to characterize the range estimation mean-square error (MSE) for both matched filtering (MF) and reciprocal filtering (RF) sensing receiver architectures. Our theoretical analysis reveals that in multi-target and clutter-rich environments, the sensing performance of the MF receiver is fundamentally limited by the fourth-order moment (kurtosis) of the constellation, which determines the data-dependent sidelobe interference level. In contrast, the RF receiver eliminates such interference at the cost of noise enhancement, with its performance governed by the inverse second-order moment of the constellation. Building on these closed-form MSE derivations, we propose a sensing-receiver specific geometric constellation shaping (GCS) framework. By jointly optimizing the constellation geometry based on the minimum Euclidean distance (MED) and receiver-dependent sensing metrics, we enable a flexible trade-off between communication reliability and sensing precision. The proposed analytical framework and GCS designs are validated through extensive numerical simulations and a hardware-based proof-of-concept experiment. Our results demonstrate that the proposed constellation shaping provides significant performance gains and facilitates a tailored sensing and communication trade-off across different receiver architectures in practical over-the-air implementations.
\end{abstract}

% Note that keywords are not normally used for peer-reviewed papers.

% For peer review papers, you can put extra information on the cover
% page as needed:
% \ifCLASSOPTIONpeerreview
% \begin{center} \bfseries EDICS Category: 3-BBND \end{center}
% \fi
%
% For peerreview papers, this IEEEtran command inserts a page break and
% creates the second title. It will be ignored for other modes.
\IEEEpeerreviewmaketitle

\section{Introduction}
% The very first letter is a 2 line initial drop letter followed
% by the rest of the first word in caps.
% 
% form to use if the first word consists of a single letter:
% \IEEEPARstart{A}{demo} file is ....
% 
% form to use if you need the single drop letter followed by
% normal text (unknown if ever used by the IEEE):
% \IEEEPARstart{A}{}demo file is ....
% 
% Some journals put the first two words in caps:
% \IEEEPARstart{T}{his demo} file is ....
% 
% Here we have the typical use of a "T" for an initial drop letter
% and "HIS" in caps to complete the first word.
\IEEEPARstart{T}he paradigm of Integrated Sensing and Communication (ISAC) represents a fundamental shift from the traditional coexistence of radar and communication systems toward unified dual-functional radar-communication (DFRC) architectures \cite{liu2022integrated, hassanien2016signaling}. By sharing a common hardware platform and spectral rescatterers, ISAC achieves significant gains in hardware and spectral efficiency while fostering synergistic interactions between sensing and communication (S\&C) functionalities \cite{luo2025isac}. Recent research in ISAC design has transitioned from radar-centric or communication-centric designs toward sophisticated joint signal processing and multiple-input multiple-output (MIMO) precoding strategies \cite{zhang2021overview, liu2020joint0, ma2021spatial, han2025next,meng2025network}. These developments have established the theoretical feasibility of ISAC and provided flexible mechanisms for balancing the inherent performance trade-offs between the two functions.

As ISAC enters the practical deployment phase of sixth-generation (6G) wireless networks, there is an intensifying demand for standard-compatible solutions that can be integrated into existing wireless frameworks \cite{liu2025sensing, prasad2004ofdm}. A prominent approach involves leveraging established reference signals, such as preambles and pilots, simultaneously for communication channel estimation and radar sensing \cite{zhu2023pilot, hua2024integrated, kumari2017ieee}. While utilizing reference signals minimizes communication overhead and integration complexity, the resulting sensing performance is fundamentally constrained by the limited time, frequency, and power rescatterers typically allocated to these signals. Consequently, reliance on reference signals alone may fail to meet the stringent accuracy and resolution requirements of emerging ISAC applications, such as high-mobility vehicular monitoring, drone detection, and environmental mapping for digital twin systems \cite{cui2021integrating}.

Beyond the exclusive use of reference signals, communication data payloads can be exploited for radar sensing, which is a concept extensively studied within the framework of communication-centric OFDM-ISAC \cite{sturm2011waveform, han2023sub, keskin2025fundamental, hu2024joint}. Traditionally, this was explored through opportunistic passive radar sensing using Digital Video Broadcasting-Terrestrial (DVB-T) signals \cite{berger2010signal, palmer2012dvb}. Subsequent studies transitioned toward mono-static configurations, demonstrating the potential of data-embedding waveforms for radar parameter estimation \cite{sturm2011waveform}. However, these early works largely treated communication signals as stochastic rather than controllable rescatterers, as data payloads are modulated by random information bits according to a certain signal constellation.

Although some early studies analyzed the impact of randomly modulated signals on radar sensing performance, recent efforts have rigorously investigated the design of communication-centric ISAC signaling by exploiting temporal and spectral degrees of freedom (DoF). The work in \cite{liu2025cp} demonstrated that an orthogonal frequency-division multiplexing (OFDM) waveform with a cyclic prefix (CP) outperforms other waveforms, such as single-carrier modulation and orthogonal time frequency space (OTFS), in terms of the ranging sidelobes of the single-input single-output (SISO) ambiguity function (AF). Moreover, the established link between communication signal properties and AF characteristics has motivated the design of pulse shaping, modulation constellations, and power allocation strategies, enabling flexible S\&C trade-offs in communication-centric ISAC systems \cite{liao2025pulse, yang2024constellation, zhang2025optimal}. 

Interestingly, the impact of the signal constellation on OFDM-based ISAC utilizing data payloads has been analyzed in terms of ranging sidelobe levels (SLLs) in \cite{liu2025uncovering}, which showed that unit-amplitude constellations, such as phase-shift keying (PSK), yield optimal sensing performance. Building on this analysis, the work in \cite{du2024reshaping} proposed probabilistic constellation shaping to flexibly balance communication mutual information (MI) and the SLLs of the auto-correlation function (ACF). Furthermore, by adopting target detection performance as a sensing metric, a joint probabilistic and geometric constellation shaping framework was introduced in \cite{geiger2025joint, geiger2025constellation}. However, sensing performance analyses based on the AF are fundamentally limited; they primarily characterize the performance of matched filtering (MF) receivers. In practice, radar receiver architectures may employ mismatched filtering (MMF) to suppress the deleterious sidelobes induced by the inherent randomness of communication data \cite{mcaulay1971optimal}.

In the context of OFDM-ISAC systems, an MMF technique known as reciprocal filtering (RF) has been employed to mitigate the range sidelobes induced by random communication payloads \cite{wojaczek2018reciprocal, rodriguez2023supervised, han2025sensing}. Although RF effectively eliminates data-dependent sidelobes through element-wise division in the frequency domain for CP-OFDM ISAC systems, the receiver simultaneously experiences a signal-to-noise ratio (SNR) loss due to noise amplification, which is known for a fundamental trade-off inherent to MMF architectures. The seminal work in \cite{wojaczek2018reciprocal} established the analytical relationship between the modulation constellation geometry and the resulting SNR loss of the RF receiver. Subsequently, the study in \cite{keskin2025fundamental} provided a comparative performance analysis of MF, RF, and linear minimum mean-square error (LMMSE) receivers, characterizing the achievable dynamic range under varying input SNR conditions. Notably, it was demonstrated that the LMMSE receiver, as a representative MMF approach, maximizes the receiver dynamic range when the input SNR is known a priori.

\begin{figure}[t!]
    \centering
    {\includegraphics[width=0.425\textwidth]{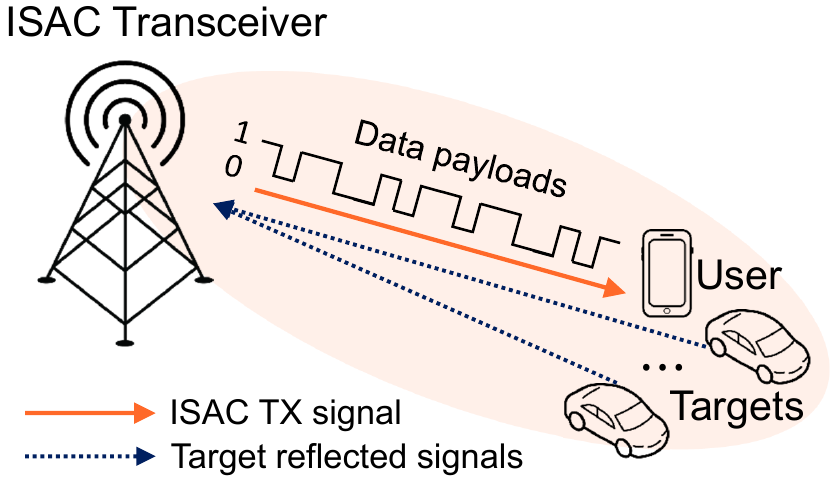}}
    \caption{Communication-centric ISAC over data payloads in multi-target environment.}
    \label{f1}
\end{figure}

Despite these advancements in characterizing OFDM-based sensing using data payloads, an explicit performance characterization for multi-target ranging remains largely unexplored. Specifically, the mechanisms, by which multiple targets and environmental clutter interfere with parameter estimation as shown in Fig. \ref{f1}, differ fundamentally between MF and RF receivers, yet a unified analytical framework for their comparison is still lacking. Furthermore, existing ISAC geometric constellation shaping strategies are primarily tailored for the conventional MF receiver \cite{xu2024experimental, hu2025learning, du2024reshaping}, which may be inherently suboptimal for RF-based processing. Consequently, there is a compelling need for a receiver-specific ISAC constellation design that accounts for the unique noise and interference characteristics of MMF architectures.

In this paper, we present a unified and explicit performance analysis of OFDM-based ISAC utilizing data payloads within an estimation-theoretic framework. Our analysis addresses receiver-specific impacts of the modulation constellation, as well as the deleterious effects of multi-target interference and clutter. Leveraging range estimation mean-square error (MSE) derivations for both MF and RF receivers, we propose a geometric constellation shaping (GCS) framework that enables flexible, receiver-specific S\&C trade-offs. \textbf{The main contributions of this paper are summarized as follows}:
\begin{itemize}
    \item We develop an estimation-theoretic sensing performance framework for CP-OFDM-based ISAC. We demonstrate that the expected Cram\'{e}r--Rao bound (CRB) over random data for target range estimation remains invariant to the choice of signal constellation. By contrast, we derive closed-form ranging MSE expressions for MF and RF receivers, explicitly quantifying the multi-target range estimation performance depending on the modulation constellation geometry.
    
    \item We analyze the robustness of MF and RF receivers against multi-target interference and clutter. Our results reveal that the range estimation performance of the MF receiver degrades in multi-target scenarios when the fourth-order moment of the constellation is larger than unity. In contrast, the RF receiver performance is governed by the inverse second-order moment of the constellation and remains independent of the number of targets or clutter scatterers.
    
    \item On this theoretical basis, we introduce a receiver-specific geometric constellation shaping framework. By employing the minimum Euclidean distance (MED) as the communication performance metric, the proposed joint optimization of constellation geometry facilitates flexible S\&C performance trade-offs specifically tailored to MF and RF sensing receiver architectures.
    
    \item The proposed analytical framework and constellation shaping strategies are validated via numerical simulations and a hardware-based proof-of-concept (PoC) experiment. Specifically, we provide experimental evidence of the trade-off between range estimation accuracy and communication throughput, thereby bridging the gap between theoretical analysis and practical over-the-air implementation.
\end{itemize}

\begin{table*}[t!]
    \centering
    \fontsize{8.3}{14}\selectfont
    \caption{Values of $\mu_4$ and $\nu_{-2}$ for PSK, $M$-QAM, and $M$-APSK modulation schemes. The APSK modulation formats with corresponding code rates are defined in the DVB standard~\cite{etsi2020dvbs2x}} 
        \label{table2}
    \begin{tabular}{>{\centering\arraybackslash} m{11em} | >{\centering\arraybackslash} m{3em}  | >
    {\centering\arraybackslash} m{3em}  | >
    {\centering\arraybackslash} m{3em}  | >{\centering\arraybackslash} m{3.5em} | >{\centering\arraybackslash} m{3.5em} | >{\centering\arraybackslash} m{3.5em} | >{\centering\arraybackslash} m{3.5em} | >{\centering\arraybackslash} m{3.5em} | >
    {\centering\arraybackslash} m{3.5em} | >{\centering\arraybackslash} m{3.5em}}
        \toprule
             & PSK & 16QAM & 32QAM & 64QAM & 128QAM & 256QAM & 16APSK (2/3) & 32APSK (2/3) & 64APSK (7/9) & 64APSK (128/180) \\ 
            \hline
        \midrule
        $\mu_4 = \frac{1}{M}\sum_{m=1}^{M} |s_m|^4$   & 1 & 1.32 & 1.31 & 1.38 & 1.34 & 1.40 & 1.25 & 1.41 & 1.64 & 1.39   \\ \hline
        $\nu_{-2} = \frac{1}{M}\sum_{m=1}^{M} |s_m|^{-2}$  & 1 & 1.89 & 2.23 & 2.69 & 2.98 & 3.44 & 2.50 & 3.23 & 2.55 & 2.13 \\ 
        \bottomrule
    \end{tabular}
\end{table*}

The remainder of this article is organized as follows. Section \ref{Sec::1} establishes the system and signal models for OFDM-based ISAC. Section \ref{Sec::RX} provides the core performance analysis, including the derivation of closed-form range estimation MSE expressions for both MF and RF processing architectures. Building on these analytical foundations, Section \ref{Sec::3} develops the receiver-specific ISAC geometric constellation shaping framework. Section \ref{Sec::4} presents extensive numerical simulation results, followed by Section \ref{Sec::5}, which validates the proposed framework through over-the-air experimental results from a PoC implementation. Finally, Section \ref{Sec::6} provides concluding remarks.

\textit{Notations}: Boldface lower-case and upper-case symbols denote vectors and matrices, respectively. $\mathbf{A} \in \mathbb{C}^{N \times M}$ and $\mathbf{B} \in \mathbb{R}^{N \times M}$ represent a complex-valued $N \times M$ matrix and a real-valued $N \times M$ matrix, respectively. The operators $(\cdot)^{T}$, $(\cdot)^{H}$, and $(\cdot)^{*}$ represent the transpose, Hermitian transpose, and conjugate, respectively. $\text{diag}(\mathbf{a})$ denotes a diagonal matrix whose diagonal elements are given by the vector $\mathbf{a}$. The operators $\odot$ and $\oslash$ represent the Hadamard (element-wise) product and the element-wise division, respectively. $\mathbb{E}[\cdot]$ denotes the statistical expectation operator. Finally, $\Re(\cdot)$ and $\Im(\cdot)$ denote the real and imaginary parts of a complex number, respectively.

\section{System Model}\label{Sec::1}
\subsection{Transmit Signal Model}
Consider an ISAC transmitter utilizing an OFDM waveform with $N$ subcarriers. The transmit signal is modulated with communication symbols drawn from a constellation set $\mathcal{S}$, denoted by the vector $\mathbf{x} = [x_0, x_1, \dots, x_{N-1}]^T$, where $x_n \in \mathcal{S}$ for $n \in \{0, 1, \dots, N-1\}$. Without loss of generality, we assume the constellation is normalized to zero mean and unit average power, such that $\mathbb{E}[|x_n|^2] = 1$. To mitigate inter-symbol interference (ISI) and maintain the orthogonality of subcarriers in multipath environments, a CP of length $N_{\text{CP}}$ is prepended to each OFDM symbol. The addition of the CP also facilitates frequency-domain processing at the sensing receiver (RX), transforming the linear convolution of the channel into a circular one, which is essential for efficient parameter estimation in OFDM radar.

The fundamental statistical properties of the constellation $\mathcal{S}$ governing the ISAC performance are defined as follows:
\begin{align}
    \mathbb{E}\left[|x_n|^4\right] = \mu_{4}, \;\;\; \mathbb{E}\left[|x_n|^{-2}\right] = \nu_{-2}, \quad \forall x_n \in \mathcal{S}, \label{eq1}
\end{align}
where $\mu_{4}$ represents the fourth-order moment (kurtosis) of the constellation, and $\nu_{-2}$ denotes the inverse second-order moment \cite{wojaczek2018reciprocal}. For an $M$-ary constellation, these moments are respectively calculated as 
\begin{align}
    \mu_{4} & = \frac{1}{M}\sum_{m=1}^{M} |x_m|^4, \\
    \nu_{-2} & = \frac{1}{M}\sum_{m=1}^{M} |x_m|^{-2}.
\end{align}

The values of $\mu_4$ and $\nu_{-2}$ for various modulation schemes, including phase-shift keying (PSK), $M$-quadrature amplitude modulation (QAM), and $M$-ary amplitude phase-shift keying (APSK), are summarized in Table \ref{table2}. Note that for unit-amplitude constellations like PSK, both $\mu_4$ and $\nu_{-2}$ are equal to unity. However, for multi-amplitude and higher-order modulation formats such as $M$-QAM and those defined in the DVB standard \cite{etsi2020dvbs2x}, these values vary significantly based on the constellation geometry. As will be demonstrated in Section \ref{Sec::RX}, these statistical parameters play a decisive role in determining the sensing performance of MF and RF receivers. This data-embedding signal serves a dual purpose, facilitating simultaneous information transfer to a communication user and radar sensing via backscattered signals from targets of interest.

\subsection{Sensing System Model}
We consider a sensing scenario involving $K$ discrete scatterers located at distinct ranges. These scatterers are assumed to be sufficiently separated in the delay domain to be resolvable, a condition generally satisfied in conventional radar systems. It is important to note that the total number of scatterers $K$ encompasses both targets of interest and environmental clutter scatterers. Furthermore, the scatterers are assumed to be slowly moving such that the Doppler shift is negligible relative to the subcarrier spacing, thereby not suffering from inter-carrier interference (ICI) \cite{keskin2021mimo}. We further assume that all targets are positioned within a range corresponding to the CP length, which ensures the absence of ISI and maintains the circularity of the channel convolution \cite{xu2025does}. 

Following CP removal and an $N$-point fast Fourier transform (FFT), the frequency-domain received signal at the sensing RX is modeled as:
\begin{align}
    \mathbf{y} = \mathbf{a}^T \mathbf{H} \mathbf{X} + \mathbf{z}, \label{RX_signal}
\end{align}
where $\mathbf{a} = [\alpha_{1}, \alpha_{2}, \dots, \alpha_{K}]^T \in \mathbb{C}^{K \times 1}$ is the complex amplitude vector incorporating the path loss and the radar cross-section (RCS) of each scatterer $k$. The target delay channel matrix is expressed as $\mathbf{H} = [\mathbf{h}(\tau_1), \mathbf{h}(\tau_2), \dots, \mathbf{h}(\tau_K)]^T \in \mathbb{C}^{K \times N}$, where $\tau_{k}$ is the round-trip time-of-flight (TOF) from the ISAC TX to scatterer $k$ and back to the receiver. The range steering vector is defined as $\mathbf{h}(\tau) = [1, e^{-j2 \pi \Delta f \tau}, \dots, e^{-j2 \pi (N-1)\Delta f \tau}]^T \in \mathbb{C}^{N \times 1}$ with subcarrier spacing $\Delta f = B/N$. The transmitted signal matrix $\mathbf{X}$ is a diagonalized representation of $\mathbf{x}$, given by $\mathbf{X} = \text{diag}(\mathbf{x})$. Finally, $\mathbf{z}$ represents the additive white Gaussian noise (AWGN) at the sensing receiver, following $\mathbf{z} \sim \mathcal{CN}(\mathbf{0}, \sigma^2 \mathbf{I}_N)$.

\section{Estimation-Theoretic Sensing Performance Analysis}\label{Sec::RX}
In this section, we characterize the delay estimation performance of OFDM-based ISAC systems under random communication signaling. We first establish a theoretical benchmark using the expected CRB and subsequently derive the MSE for specific receiver architectures, accounting for the impact of the signal constellation.

\subsection{Expected CRB for Delay Estimation}
Assuming the targets are sufficiently resolved in the delay domain, we analyze the likelihood function for a single target to provide a tractable basis for investigating how specific constellation geometry affects ranging estimation accuracy. Under the assumption of AWGN, the log-likelihood function for the delay parameter $\tau$ of the $k$-th target is given by
\begin{equation}
    \mathcal{L}(\tau_k) = -\frac{1}{\sigma^2} \left\| \mathbf{y} - \alpha_k \mathbf{h}(\tau_k)^T \mathbf{X} \right\|^2 + c_0,
\end{equation}
where $c_0$ is a constant term independent of the unknown parameters. For a single-target scenario or well-separated multiple targets, the Fisher information (FI) regarding the delay $\tau_k$ is obtained by taking the second derivative of the log-likelihood with respect to $\tau_k$. We first define the signal model for the $k$-th target as $\mathbf{s}(\tau_k) = \alpha_k \mathbf{h}(\tau_k)^T \mathbf{X}$. Then, it is given by
\begin{equation}
    \mathcal{I}(\tau_k) = \frac{2}{\sigma^2} \Re \left\{ \left( \frac{\partial \mathbf{s}(\tau_k)}{\partial \tau_k} \right)^H \left( \frac{\partial \mathbf{s}(\tau_k)}{\partial \tau_k} \right) \right\}. \label{FI}
\end{equation}
The partial derivative of the signal vector with respect to the delay $\tau_k$ simplifies to the vector $\dot{\mathbf{s}}(\tau_k)$ with elements $s'_n = -j2\pi n \Delta f \alpha_k e^{-j2\pi n \Delta f \tau_k} x_n$. Substituting this into \eqref{FI}, we obtain
\begin{align}
    \mathcal{I}(\tau_k) &= \frac{2}{\sigma^2} \sum_{n=0}^{N-1} | -j2\pi n \Delta f \alpha_k e^{-j2\pi n \Delta f \tau_k} x_n |^2 \\
    &=  \frac{8 \pi^2 \Delta f^2 |\alpha_k|^2}{\sigma^2} \sum_{n=0}^{N-1} n^2 |x_n|^2.
\end{align}

The CRB, which provides a lower bound on the variance of any unbiased estimator, is defined as the inverse of the FI:
\begin{equation}
\mathrm{CRB}_{\tau_k} = \frac{1}{\mathcal{I}(\tau_k)}= \frac{\sigma^2}{8 \pi^2 \Delta f^2 |\alpha_k|^2 \sum_{n=0}^{N-1} n^2 |x_n|^2}. \label{CRB}
\end{equation}
Note that \eqref{CRB} is a random variable that depends on the instantaneous realization of the communication symbols $x_n$. To provide a fundamental sensing performance limit for the ISAC system, we define the expected CRB under random signaling in the following theorem. It is important to note that the instantaneous CRB in \eqref{CRB} is itself a random quantity, as it depends on the specific realization of the randomly modulated data symbols. In general, the expectation of the CRB is not equal to the CRB evaluated at the expected FI, due to the nonlinearity introduced by the inversion operation.

\begin{theorem}\label{theo1}
The expected CRB for delay estimation in an OFDM-based ISAC system, modulated by a constellation with kurtosis $\mu_4$, is approximately expressed as
\begin{align}
    \mathbb{E}[\mathrm{CRB}_{\tau_k}] \approx \frac{\sigma^2}{8 \pi^2 \Delta f^2 |\alpha_k|^2}
    \left( \frac{3}{N^3} + \frac{27(\mu_4 - 1)}{5N^4} \right). \label{Theo_eq0}
\end{align}
\end{theorem}
\renewcommand\qedsymbol{$\blacksquare$}
\begin{proof}
Please refer to Appendix \ref{proof_theo1}.
\end{proof}

\noindent \textbf{Remark 1}: From Theorem \ref{theo1}, it is observed that the constellation geometry, characterized by $\mu_4$, appears in the second term within the brackets of \eqref{Theo_eq0}. This term vanishes as $N$ increases, leaving the expected CRB to be governed by the first term, which is constellation-invariant. This suggests that the CRB alone may not fully capture the practical performance degradation induced by random data payloads. However, it provides a fundamental insight: theoretically, an optimal delay estimator can achieve a performance limit nearly independent of the constellation, provided the estimator is properly designed. Nevertheless, we show below that this insight does not capture the realistic performance of practical sensing receivers such as MF and RF. To further investigate the practical limitations, we evaluate the MSE of delay estimation for these specific ISAC receiver architectures in the following section.

\subsection{Delay Estimation MSE with Matched Filtering Receiver}
In this subsection, we derive the estimation error for the delay $\tau_k$ of the $k$-th scatterer under MF processing. At the sensing RX, the MF operation is performed by multiplying the received signal $\mathbf{y}$ by the conjugate of the transmitted symbols $\mathbf{X}^H$. The resulting output, $\mathbf{y}_{\text{MF}} = \mathbf{y}\mathbf{X}^H$, is given by
\begin{align}
    \mathbf{y}_{\text{MF}} = \mathbf{a}^T \mathbf{H} |\mathbf{X}|^2 + \mathbf{z}_{\text{MF}},
\end{align}
where $\mathbf{z}_{\text{MF}} = \mathbf{z}\mathbf{X}^H$. Given the unit-variance assumption of the constellation, $\mathbb{E}[|x_n|^2]=1$, the filtered noise $\mathbf{z}_{\text{MF}}$ maintains the same statistical characteristics as the original AWGN vector $\mathbf{z}$. 

We consider a general delay estimator that selects atoms from a dictionary $\mathcal{A}=\{\mathbf{h}(\tau) \; | \; \tau \in \mathcal{T}\}$, where $\mathcal{T}$ represents a delay candidate set with sufficient resolution. A common and computationally efficient implementation of such an estimator is the inverse discrete Fourier transform (IDFT). The delay estimate is formally defined as
\begin{align} \label{estimator}
    \hat{\tau} = \arg \max_{\tau \in \mathcal{T}} \left| \mathbf{h}^H(\tau) \mathbf{y}_{\text{MF}}^T \right|.
\end{align}
To analyze the estimator performance, let $s_{\text{MF}}(\tau) = \mathbf{h}^H(\tau) \mathbf{y}_{\text{MF}}^T$ denote the correlation output, which can be expanded as
\begin{equation}
    s_{\text{MF}}(\tau) = \sum_{k=1}^K \sum_{n=0}^{N-1} \alpha_k |x_n|^2 e^{j2\pi n \Delta f (\tau - \tau_k)} + \sum_{n=0}^{N-1} x_n^* z_n e^{j2\pi n \Delta f \tau}.
\end{equation}
The objective function for the estimation is defined as $f(\tau) = |s_{\text{MF}}(\tau)|^2$. The following lemma provides a general approximation for the delay estimation error in terms of the derivatives of $f(\tau)$ around the true delay.

\begin{lemma} \label{lemma1}
    Assuming that the $k$-th scatterer is the target of interest and $\tau_k$ represents its true delay, the estimation error under the assumption of a high SNR can be approximated using the first-order optimality condition as
    \begin{equation}
        \hat{\tau}_k - \tau_k = - \frac{\dot{f}(\tau_k)}{\ddot{f}(\tau_k)}, \label{MSE1}
    \end{equation}
    where $\dot{f}(\tau_k)$ and $\ddot{f}(\tau_k)$ denote the first and second derivatives of $f(\tau_k)$ in terms of $\tau_k$.
\end{lemma}

\begin{proof}
    Assuming that the objective function $f(\tau)$ exhibits a distinct and sharp global maximum in the vicinity of the true delay $\tau_k$, a second-order Taylor expansion of $f(\tau)$ around $\tau_k$ is given by:
    \begin{align}
        f(\tau) \approx f(\tau_k) + \dot{f}(\tau_k)(\tau - \tau_k) + \frac{1}{2} \ddot{f}(\tau_k)(\tau - \tau_k)^2. \label{Taylor}
    \end{align}
    Differentiating \eqref{Taylor} with respect to $\tau$ yields:
    \begin{align}
        \dot{f}(\tau) \approx \dot{f}(\tau_k) + \ddot{f}(\tau_k)(\tau - \tau_k). \label{Taylor2}
    \end{align}
    At the estimate $\hat{\tau}_k$, the first-order optimality condition requires $\dot{f}(\hat{\tau}_k) = 0$. By evaluating \eqref{Taylor2} at $\tau = \hat{\tau}_k$ and solving for the estimation error $\hat{\tau}_k - \tau_k$, we arrive at the expression in \eqref{MSE1}, which completes the proof.
\end{proof}
The expression in Lemma \ref{lemma1} indicates that the estimation error is fundamentally governed by two factors: the fluctuation power of the first derivative $\dot{f}(\tau_k)$, arising from additive noise and mutual interference from the sidelobes of other scatterers, and the mainlobe sharpness, characterized by the curvature $\ddot{f}(\tau_k)$. We evaluate the resulting MSE by taking the expectation of the squared error in \eqref{MSE1} as:
\begin{equation} \label{MSE2}
    \mathbb{E}[(\hat{\tau}_k - \tau_k)^2] = \mathbb{E}\left[ \left( \frac{\dot{f}(\tau_k)}{\ddot{f}(\tau_k)} \right)^2 \right] 
    \approx \frac{\mathbb{E}[|\dot{f}(\tau_k)|^2]}{|\mathbb{E}[\ddot{f}(\tau_k)]|^2}.
\end{equation}
It is worth to note that this approximation is justified in the high-SNR regime, where the fluctuations of the second derivative $\ddot{f}(\tau_k)$ are negligible compared to its expected value, allowing for the decoupling of the numerator and denominator expectations. This approximation, which effectively assumes that the second derivative is highly concentrated around its mean, is widely adopted in the performance analysis of nonlinear estimators \cite{van2002optimum}. Such a statistical linearization is valid when the objective function exhibits a stable curvature and the estimator operates within a local neighborhood of the true parameter under sufficiently high SNR conditions.

To relate the MSE expression in \eqref{MSE2} to the correlation output $s_{\text{MF}}(\tau)$, we introduce the following lemma:
\begin{lemma} \label{lemma2}
    The delay estimation MSE for the $k$-th scatterer as defined in \eqref{MSE2} can be reformulated in terms of the derivatives of the correlation output as
    \begin{equation}
        \mathbb{E}[(\hat{\tau}_k - \tau_k)^2] = \frac{\mathbb{E}[|\dot{s}_{\text{MF}}(\tau_k)|^2]}{2|\mathbb{E}[\ddot{s}_{\text{MF}}(\tau_k)]|^2}. \label{MSE3}
    \end{equation}
\end{lemma}

\begin{proof}
    The first derivative of the objective function $f(\tau) = |s_{\text{MF}}(\tau)|^2$ at the true delay $\tau_k$ is given by
    \begin{align}
        \dot{f}(\tau_k) = 2 \Re\{\dot{s}_{\text{MF}}(\tau_k) s_{\text{MF}}^*(\tau_k)\}.
    \end{align}
    Under the assumption of a high SNR and large $N$, the correlation peak at the true delay is dominated by the target signal component, such that $s_{\text{MF}}^*(\tau_k) \approx N\alpha_k^*$. Then, the expectation of the squared derivative is given by
    \begin{align}
        \mathbb{E}[|\dot{f}(\tau_k)|^2] &\approx 4N^2|\alpha_k|^2 \mathbb{E}\left[\left(\Re\{\dot{s}_{\text{MF}}(\tau_k)\}\right)^2\right] \nonumber \\
        &= 2N^2|\alpha_k|^2 \mathbb{E}[|\dot{s}_{\text{MF}}(\tau_k)|^2],
    \end{align}
    where we utilized the property that for a complex variable $Z$, $\mathbb{E}[(\Re\{Z\})^2] = \frac{1}{2}\mathbb{E}[|Z|^2]$ under the uniform phase distribution. 
    
    Similarly, the second derivative of $f(\tau)$ is expressed as
    \begin{align}
        \ddot{f}(\tau_k) &= 2 \mathrm{Re}\{\ddot{s}_{\text{MF}}(\tau_k) s_{\text{MF}}^*(\tau_k) + |\dot{s}_{\text{MF}}(\tau_k)|^2\}. 
    \end{align}
    Taking the expectation and squaring the result yields
    \begin{align}
        |\mathbb{E}[\ddot{f}(\tau_k)]|^2 \approx 4N^2|\alpha_k|^2 |\mathbb{E}[\ddot{s}_{\text{MF}}(\tau_k)]|^2.
    \end{align}
    Substituting these results into \eqref{MSE2} yields the expression in \eqref{MSE3}, which completes the proof.
\end{proof}

Building on Lemma \ref{lemma2}, we now derive a closed-form expression for the delay estimation MSE of the MF receiver. The first derivative of $s_{\text{MF}}(\tau)$ evaluated at the true delay $\tau_k$ is given by
\begin{align} \label{firstDev1}
    \dot{s}_{\text{MF}}(\tau_k) = (j2\pi \Delta f) \bigg( & \underbrace{\sum_{j \ne k}^K \sum_{n=0}^{N-1} n \alpha_j |x_n|^2 e^{j2\pi n \Delta f (\tau_k - \tau_j)}}_{\text{Sidelobe interference from other scatterers}, \;\mathrm{C}_{1}} \nonumber \\ 
    & + \underbrace{\sum_{n=0}^{N-1} n x_n^* z_n e^{j2\pi n \Delta f \tau_k}}_{\text{Noise}, \;\mathrm{C}_{2}} \bigg),
\end{align}
where the term corresponding to the desired target ($j=k$) vanishes because the target mainlobe reaches its extremum at $\tau_k$ and its first derivative goes to zero. Thus, the first derivative is purely a function of the sidelobe interference from other scatterers and the receiver noise. 

Similarly, the second derivative of $s_{\text{MF}}(\tau)$ at $\tau_k$ is expressed as
\begin{align} \label{SecondDev1}
    \ddot{s}_{\text{MF}}(\tau_k) = (j2\pi\Delta f)^{2} \bigg( & \sum_{j=1}^K \sum_{n=0}^{N-1} n^2 \alpha_j |x_n|^2 e^{j2\pi n\Delta f(\tau_k-\tau_j)} \nonumber \\
    & + \sum_{n=0}^{N-1} n^2 x_n^* z_n e^{j2\pi n\Delta f\tau_k} \bigg).
\end{align}
Utilizing the statistical properties of the constellation and the noise, we arrive at the following theorem for the delay estimation MSE under the MF receiver processing.

\begin{theorem}\label{theo2}
The closed-form expression for the delay estimation MSE of the $k$-th scatterer using the MF receiver is given by
\begin{align}
    \mathrm{MSE}_{\text{MF},k} = \frac{3 \left( (\mu_{4}-1) \sum_{j\neq k}^K |\alpha_{j}|^2 + \sigma^2 \right)} {8\pi^2\Delta f^2 |\alpha_{k}|^2 N^3}. \label{Theo_eq2}
\end{align}
\end{theorem}
\renewcommand\qedsymbol{$\blacksquare$}
\begin{proof}
Please refer to Appendix \ref{proof_theo2}.
\end{proof}

\noindent \textbf{Remark 2 (Impact of Modulation Constellation and Multi-target Interference for MF Receiver):} It is observed that the delay estimation MSE for the conventional MF receiver is fundamentally limited by the sidelobes of other delay scatterers. The power of these data-dependent sidelobes is directly scaled by the constellation kurtosis through $\mu_4 - 1$. This finding is consistent with the AF analysis in \cite{liu2025uncovering}, which demonstrates that the average sidelobe level of the AF under random signaling is proportional to the fourth-order moment of the constellation. Furthermore, comparing Theorem \ref{theo1} and Theorem \ref{theo2} reveals that the CRB fails to capture the performance degradation induced by multi-target interference under random signaling. Specifically, \eqref{Theo_eq2} implies that MF-based estimators, including FFT-based and subspace-based methods, are generally non-efficient in multi-target scenarios, $K > 1$. These estimators achieve the CRB only when a unit-amplitude constellation, $\mu_4=1$, is employed or in the strictly single-target case , $K=1$, where the interference term vanishes.

\subsection{Delay Estimation MSE with Reciprocal Filtering Receiver}
We now investigate a specific class of mismatched filtering (MMF) receiver designed to suppress the data-dependent sidelobes of the AF at the cost of SNR degradation. Specifically, we focus on reciprocal filtering (RF), an approach that equalizes the transmitted communication data in ISAC signaling \cite{sturm2011waveform}. The RF receiver is implemented via element-wise division of the received signal $\mathbf{y}$ by the transmit vector $\mathbf{x}$, yielding the output $\mathbf{y}_{\text{RF}} = \mathbf{y} \oslash \mathbf{x}$, which is expressed as
\begin{align}
    \mathbf{y}_{\text{RF}} = \mathbf{a}^T \mathbf{H} + \mathbf{z}_{\text{RF}}.
\end{align}
In contrast to the MF output, $\mathbf{y}_{\text{RF}}$ is independent of the random modulation in the signal term, which effectively eliminates the sidelobe interference from other delay scatterers. However, while $\mathbf{z}_{\text{RF}}$ remains zero-mean, its variance is reshaped by the RF operation as
\begin{align}
    \mathrm{Var}(z_{\text{RF},n}) = \mathbb{E}\!\left[\left|z_n x_n^{-1}\right|^2\right] = \sigma^2 \mathbb{E}\bigl[|x_n|^{-2}\bigr] = \sigma^2 \nu_{-2}.
\end{align}
Here, it should be noted that the data payload is independent between subcarriers. Accordingly, this indicates that an SNR loss at RF output is proportional to the inverse second-order moment $\nu_{-2}$ of the constellation. 

Analogous to the MF case, the delay estimator for the RF receiver is defined as
\begin{align} \label{estimator_RF}
    \hat{\tau} = \arg \max_{\tau \in \mathcal{T}} \left| \mathbf{h}^H(\tau) \mathbf{y}_{\text{RF}}^T \right|,
\end{align}
where $\mathbf{h}(\tau) \in \mathcal{A}$. Let $s_{\text{RF}}(\tau) = \mathbf{h}^H(\tau) \mathbf{y}_{\text{RF}}^T$ denote the correlation output, which is given by
\begin{equation}
    s_{\text{RF}}(\tau) = \sum_{k=1}^K \sum_{n=0}^{N-1} \alpha_k e^{j2\pi n \Delta f (\tau - \tau_k)} + \sum_{n=0}^{N-1} z_{\text{RF},n} e^{j2\pi n \Delta f \tau}.
\end{equation}
By applying the estimation-theoretic framework established in Lemma \ref{lemma1} and Lemma \ref{lemma2}, we derive the following closed-form expression for the delay estimation MSE of the RF receiver.

\begin{theorem}\label{theo3}
The closed-form expression for the delay estimation MSE of the $k$-th target using the RF receiver is given by
\begin{align}
    \mathrm{MSE}_{\text{RF},k} = \frac{3 \sigma^2 \nu_{-2}}{8\pi^2\Delta f^2 |\alpha_{k}|^2 N^3}. \label{Theo_eq3}
\end{align}
\end{theorem}
\renewcommand\qedsymbol{$\blacksquare$}
\begin{proof}
Please refer to Appendix \ref{proof_theo3}.
\end{proof}
\noindent \textbf{Remark 3 (Impact of Modulation Constellation and Multi-target Interference for RF Receiver):} Unlike the delay estimation MSE of the MF receiver, the RF receiver's MSE is invariant to the presence of other delay scatterers. This immunity arises because the element-wise division by the communication symbols perfectly equalizes the signal term, effectively transforming the sidelobes into a noise amplification effect. Specifically, the RF receiver incurs an SNR loss proportional to the inverse second-order moment of the constellation, $\nu_{-2}$. Consequently, comparing Theorem \ref{theo1} and Theorem \ref{theo3} shows that the RF receiver is inherently unable to achieve the CRB even under high SNR conditions and a single-target scenario for non-unit-amplitude constellations where $\nu_{-2} > 1$.

Building on the MSE analysis for MF and RF receivers, we establish the following corollary regarding their relationship with the theoretical CRB.
\begin{corollary}\label{corr1}
    For a multi-target scenario ($K > 1$), the delay estimation MSE of MF and RF receivers, employing conventional estimators based on the dictionary $\mathcal{A}=\{\mathbf{h}(\tau) \; | \; \tau \in \mathcal{T}\}$, asymptotically achieves the expected CRB if and only if the communication data is modulated by a unit-amplitude constellation, such as PSK.
\end{corollary}

\begin{proof}
    In the high-SNR regime, a comparison of \eqref{Theo_eq2} and \eqref{Theo_eq3} with the expected CRB in \eqref{Theo_eq0} reveals that $\mathrm{MSE}_{\text{MF}} = \mathrm{MSE}_{\text{RF}} = \mathbb{E}[\mathrm{CRB}]$ holds if and only if $\mu_4 = 1$ and $\nu_{-2} = 1$. Under the normalization constraint $\mathbb{E}[|x_n|^2] = 1$, these statistical conditions are satisfied exclusively by unit-amplitude constellations, where $|x_n| = 1$ for all $x_n \in \mathcal{S}$.
\end{proof}
While the preceding analysis demonstrates that both MF and RF sensing receivers achieve the best performance with unit-amplitude constellations, such geometries are often suboptimal for high-order communication throughput. This performance gap motivates a joint design approach for ISAC constellation shaping. Specifically, by leveraging the derived closed-form MSE expressions, we can formulate a geometric constellation shaping framework that yields a flexible trade-off between sensing accuracy and communication reliability. In the following section, we propose a receiver-specific optimization problem to tailor the constellation geometry to flexibly balance the S\&C performance.

\section{Receiver-Specific ISAC Geometric Constellation Shaping}\label{Sec::3}
Building on the analytical MSE expressions derived for MF and RF architectures, we propose a geometric constellation shaping (GCS) framework designed to facilitate a flexible trade-off between sensing accuracy and communication reliability. Unlike probabilistic shaping, which modifies the prior probabilities of symbols, GCS optimizes the coordinates of the symbols $\{s_m\}_{m=1}^M$ in the complex plane to directly influence the statistical moments $\mu_4$ and $\nu_{-2}$.

\subsection{Communication Performance Metric: Minimum Euclidean Distance}
In communication systems, the reliability of the link is fundamentally characterized by the bit-error rate (BER). Under AWGN conditions at high SNR, the BER performance is dominated by the most likely error events, specifically the confusion between the two closest points in the constellation set $\mathcal{S}$ \cite{caire2002bit}. This is captured by the minimum Euclidean distance (MED), defined as
\begin{align}
    d_{\min} = \min_{s_i \neq s_j \in \mathcal{S}} |s_i - s_j|.
\end{align}
Maximizing $d_{\min}$ under a normalized power constraint maximizes the noise immunity of the constellation. In our design, $d_{\min}$ acts as the primary communication metric, competing with sensing-centric objectives that tend to push symbols toward the unit circle.

\begin{figure*}[t!]
    \centering
    \subfigure[]{\includegraphics[width=0.30\textwidth]{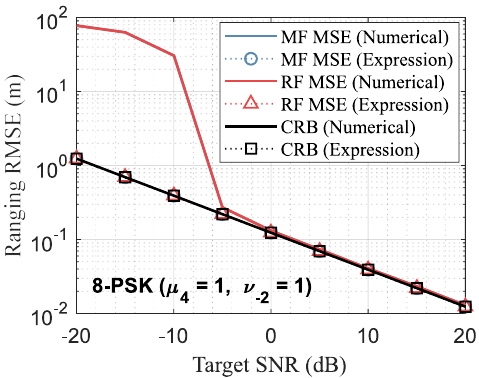}}
    \subfigure[]{\includegraphics[width=0.30\textwidth]{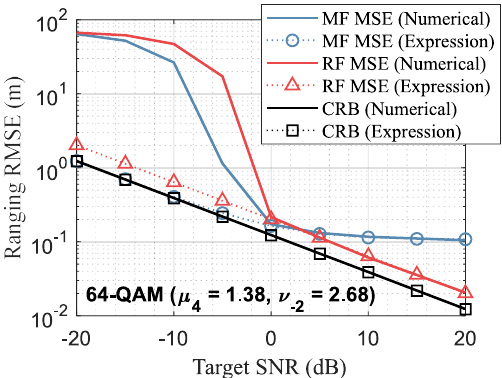}}
    \subfigure[]{\includegraphics[width=0.30\textwidth]{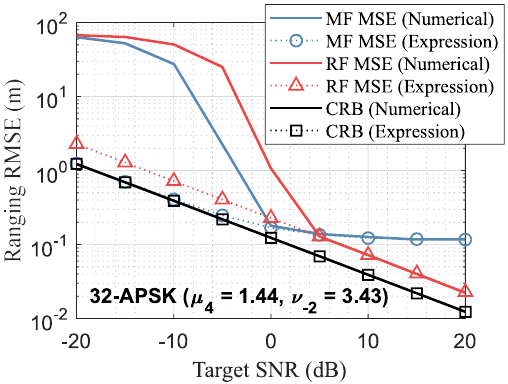}} \\
    \subfigure[]{\includegraphics[width=0.30\textwidth]{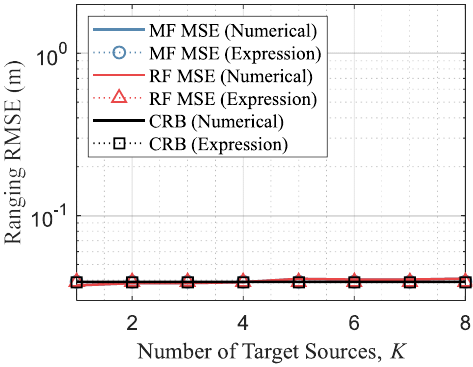}}
    \subfigure[]{\includegraphics[width=0.30\textwidth]{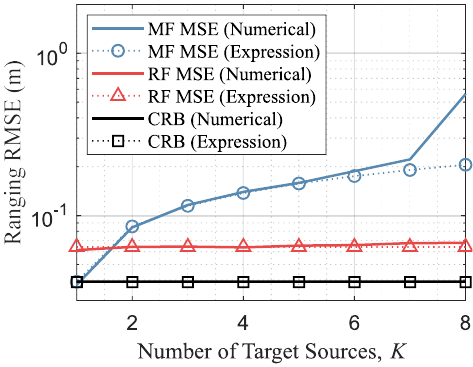}}
    \subfigure[]{\includegraphics[width=0.30\textwidth]{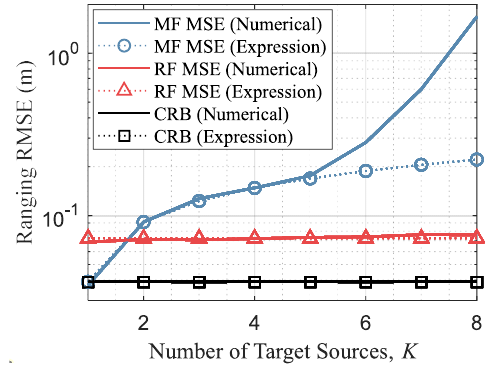}}
    \caption{Range estimation accuracy under three different modulation constellations. Subplots (a), (b), and (c) show the range estimation root-MSE (RMSE) versus target SNR with $K = 3$, while (d), (e), and (f) present the RMSE versus the number of targets at 10~dB SNR. (a) and (d) for QPSK, (b) and (e) for 64QAM, and (c) and (f) for customized 32APSK, respectively.}
    \label{Fig::2}
\end{figure*}

\subsection{Optimization Problem Formulation}
We formulate the joint ISAC constellation design as a multi-objective optimization problem. The goal is to minimize a receiver-specific sensing penalty $f_s(\mathcal{S})$ while simultaneously maximizing the communication-centric MED. The problem is formulated as:
\begin{subequations}\label{Eqn::P1}
    \begin{align}
    \min_{\{s_m\}_{m=1}^M} \quad & \rho f_s(\mathcal{S}) + (1-\rho) (-d_{\min}) \label{P1_obj} \\
    \text{subject to} \quad & |s_i - s_j| \geq d_{\min}, \quad \forall i \neq j, \label{P1_med} \\
    & \frac{1}{M}\sum_{m=1}^M s_m = 0, \label{P1_mean} \\
    & \frac{1}{M}\sum_{m=1}^M s_m^2 = 0, \label{P1_symmetry} \\
    & \frac{1}{M}\sum_{m=1}^M |s_m|^2 = 1, \label{P1_power}
    \end{align}
\end{subequations}
where $\rho \in [0,1]$ is the priority weighting factor between sensing and communication functionality. The sensing objective $f_s(\mathcal{S})$ is selected based on the receiver architecture:
\begin{itemize}
    \item \textbf{MF Receiver:} We set $f_s(\mathcal{S}) = \mu_4 = \frac{1}{M}\sum_{m=1}^{M} |x_m|^4$. Minimizing the kurtosis suppresses the data-dependent sidelobes in the ambiguity function, effectively steering the constellation geometry toward a PSK-like structure.
    \item \textbf{RF Receiver:} We set $f_s(\mathcal{S}) = \nu_{-2} = \frac{1}{M}\sum_{m=1}^{M} |x_m|^{-2}$. Minimizing the inverse second moment mitigates the noise amplification inherent in reciprocal filtering, preventing constellation points from being placed too close to the origin.
\end{itemize}
The constraints in \eqref{P1_mean}--\eqref{P1_power} ensure the fundamental statistical properties of the constellation. Specifically, \eqref{P1_mean} ensures the zero-mean constraint. The symmetry constraint in \eqref{P1_symmetry} imposes circularly symmetric condition, which ensures that the real and imaginary parts of the symbols are balanced. Finally, \eqref{P1_power} enforces the unit-power constraint, normalizing the average energy of the constellation to one.
Since \eqref{Eqn::P1} involves non-convex distance and power constraints, we employ sequential quadratic programming (SQP) to find efficient local optima. To ensure the discovery of a high-quality global solution, we utilize a multi-start strategy by initializing the SQP algorithm with a diverse set of random constellations and selecting the design that yields the minimum objective value.

\section{Numerical Simulations} \label{Sec::4}
In this section, we validate the derived theoretical framework, including the closed-form MSE expressions and the receiver-specific GCS framework, through comprehensive numerical simulations. We consider a CP-OFDM system with $N = 256$ subcarriers, a total bandwidth of $B = 50$ MHz, and a CP duration of $0.64~\mu$s. Each simulation result is obtained by averaging over $1000$ independent Monte Carlo trials. Each target SNR is defined as $|\alpha_k|^2/\sigma^2$. For delay estimation, we employ a high-resolution subspace-based matrix pencil (MP) estimator following the initial MF and RF processing stages \cite{sarkar1995using}.

\subsection{CRB and MSE Performance Analysis}
To investigate the relationship between the modulation constellation geometry and sensing performance, we evaluate four distinct constellations: QPSK, 64QAM, and a non-standard 32APSK configuration featuring four concentric rings with a radius ratio of $1:2:3:4$, whose $\mu_4=1.44$ and $\nu_{-2}=3.43$.

First, we examine the range estimation accuracy across various constellations to verify the validity of Theorem \ref{theo2} and Theorem \ref{theo3}. As illustrated in Fig. \ref{Fig::2}, the expected CRB remains largely invariant to the specific modulation format, consistent with Theorem \ref{theo1}. In contrast, the Root-MSE (RMSE) of the MF- and RF-based receivers exhibits clear dependency on the constellation moments $\mu_4$ and $\nu_{-2}$, respectively. Notably, the derived closed-form MSE expressions show excellent agreement with the numerical results for target SNRs exceeding $0$~dB. An important observation from Fig. \ref{Fig::2}(a)--(c) is that the MF receiver performance reaches an error floor as the SNR increases. This saturation is attributed to the data-dependent sidelobe interference from competing targets, which becomes the dominant error scatterer over thermal noise in multi-target scenarios ($K > 1$). Conversely, the RF-based receiver avoids this floor by effectively equalizing the sidelobe effects, maintaining a performance trajectory parallel to the CRB.

The influence of the number of scatterers $K$ on range estimation is detailed in Fig. \ref{Fig::2}(d)--(f). For QPSK, which serves as the sensing-optimal baseline ($\mu_4 = \nu_{-2} = 1$), the RMSE remains constant regardless of $K$. However, for non-unit-amplitude constellations ($\mu_4 > 1$), the MF receiver suffers from progressive performance degradation as $K$ increases, due to the cumulative power of the sidelobes identified in \eqref{Theo_eq2}. In contrast, the RF receiver maintains constant performance irrespective of the scatterer count, though it sustains a persistent gap from the CRB. This gap is the direct result of the noise amplification factor $\nu_{-2}$ derived in \eqref{Theo_eq3}. These results confirm that our closed-form expressions provide a tractable and accurate characterization of sensing performance in communication-centric OFDM-ISAC systems.

\subsection{Performance Trade-offs in ISAC GCS}
\begin{figure}[t!]
    \centering
    \subfigure[]{\includegraphics[width=0.48\textwidth]{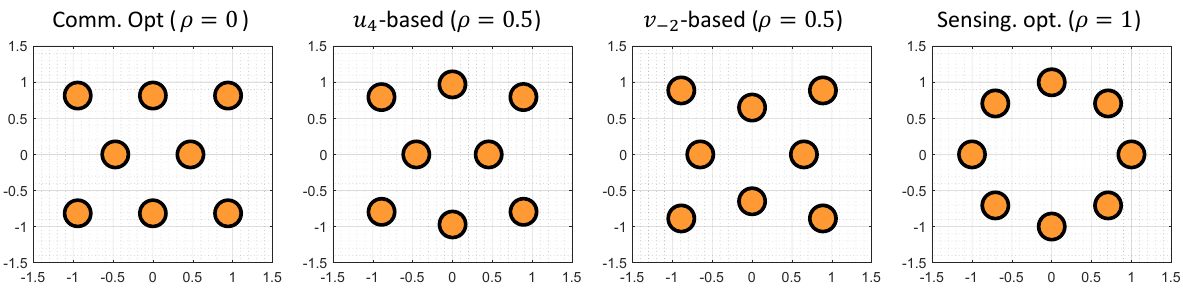}} \\
    \subfigure[]{\includegraphics[width=0.48\textwidth]{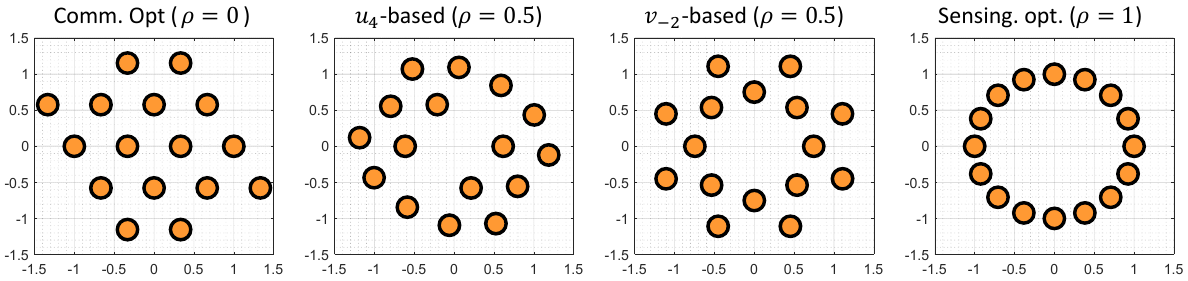}}
    \caption{Designed receiver-specific ISAC constellation geometries with various priority weights when (a) $M = 8$ and (b) $M = 16$.}
    \label{Fig::3}
\end{figure}

\begin{table}[t!]
    \centering
    \fontsize{8.3}{12}\selectfont
    \caption{Values of $d_{\min}$, $\mu_4$ and $\nu_{-2}$ for the resulting 8GCS and 16GCS ISAC constellations.} 
        \label{table3}
    \begin{tabular}{>{\centering\arraybackslash} m{3.5em} | >{\centering\arraybackslash} m{4em}  | >
    {\centering\arraybackslash} m{5em}  | >
    {\centering\arraybackslash} m{5em}  | >{\centering\arraybackslash} m{4em}}
        \toprule
             $M=8$ & $\rho = 0$ & $\rho = 0.5$, ($\mu_4$-based) & $\rho = 0.5 $ ($\nu_{-2}$-based) & $\rho = 1$  \\ 
            \hline
        \midrule
        $d_{\min}$   & 0.94 & 0.91 & 0.92 & 0.77   \\ \hline
        $\mu_4$   & 1.33 & 1.25 & 1.33 & 1.00    \\ \hline
        $\nu_{-2}$  & 1.82 & 1.83 & 1.5 & 1.00 \\ 

        \toprule
             $M=16$ & $\rho = 0$ & $\rho = 0.5$, ($\mu_4$-based) & $\rho = 0.5 $ ($\nu_{-2}$-based) & $\rho = 1$  \\ 
            \hline
        \midrule
        $d_{\min}$   & 0.67 & 0.58 & 0.58 & 0.39   \\ \hline
        $\mu_4$   & 1.37 & 1.15 & 1.19 & 1.00    \\ \hline
        $\nu_{-2}$  & 2.18 & 1.30 & 1.23 & 1.00 \\ 
        \bottomrule
    \end{tabular}
\end{table}

\begin{figure}[t!]
    \centering
    \subfigure[]{\includegraphics[width=0.48\textwidth]{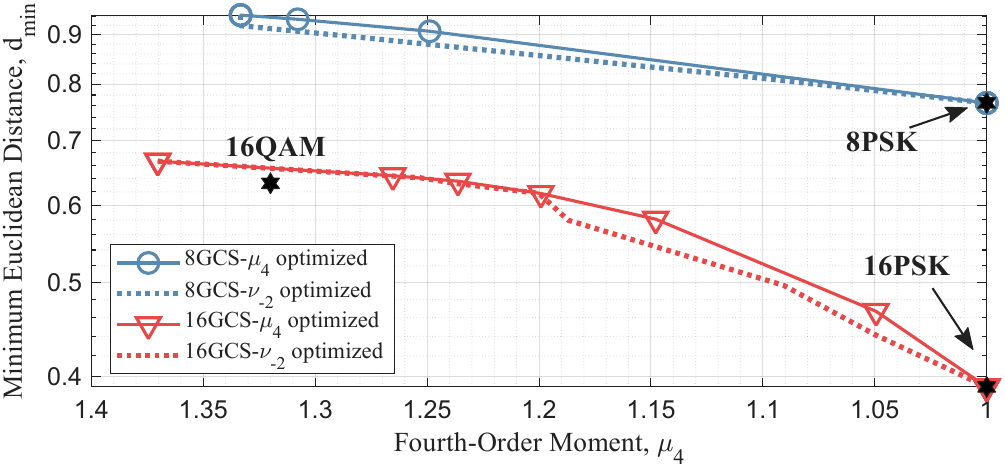}} \\
    \subfigure[]{\includegraphics[width=0.48\textwidth]{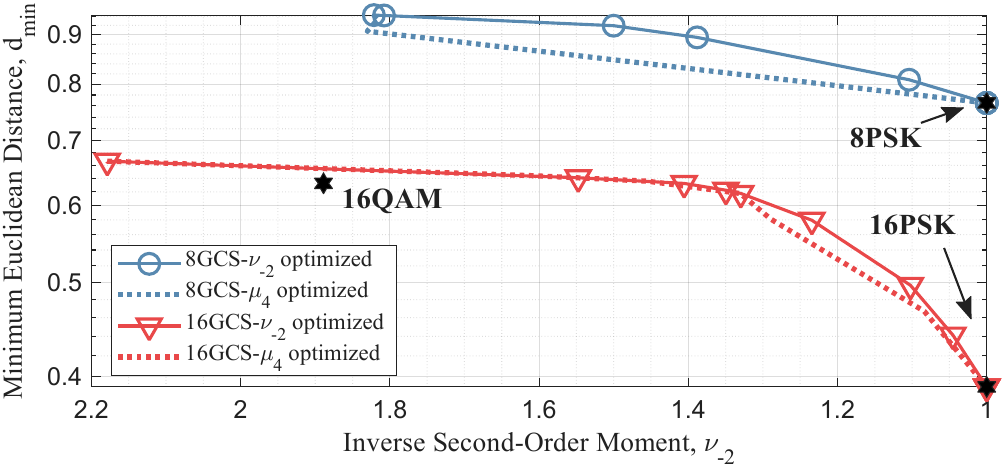}}
    \caption{Performance trade-offs in ISAC GCS: (a) The forth-order moment ($\mu_4$) vs. MED ($d_{\min}$), and (b) the inverse second-order moment ($\nu_{-2}$) vs. MED ($d_{\min}$).}
    \label{Fig::4}
\end{figure}
We further evaluate the proposed receiver-specific GCS framework in terms of its capability to facilitate a flexible S\&C performance trade-off. Fig. \ref{Fig::3}(a) and \ref{Fig::3}(b) illustrate the 8GCS and 16GCS solutions obtained by solving the optimization problem in \eqref{Eqn::P1} for various values of $\rho$. In the extreme cases of communication-optimal ($\rho = 0$) and sensing-optimal ($\rho = 1$) designs, the resulting constellations converge to identical geometries regardless of the sensing receiver architecture. 

However, in the joint-optimization regime ($0 < \rho < 1$), the GCS solutions diverge significantly based on the sensing receiver processing. As shown in Fig. \ref{Fig::3}, even for similar $d_{\min}$ values, which imply equivalent communication reliability, the constellations tailored for MF and RF receivers exhibit distinct spatial distributions. This behavior is driven by the different statistical sensitivities: the MF-oriented GCS seeks to minimize magnitude fluctuations, $\mu_4$, while the RF-oriented GCS avoids placing symbols near the origin to prevent noise amplification, $\nu_{-2}$. The specific values for $d_{\min}$, $\mu_4$, and $\nu_{-2}$ for these designs are summarized in Table \ref{table3}.

The fundamental trade-offs between communication metric ($d_{\min}$) and sensing metrics ($\mu_4$ and $\nu_{-2}$) are depicted in Fig. \ref{Fig::4}. The results underscore a critical design insight: a constellation optimized for the MF receiver exhibits significant performance loss if processed by an RF receiver, and vice versa. Specifically, a low-$\mu_4$ constellation does not necessarily yield a low $\nu_{-2}$, particularly if the design includes points near the origin that maximize $d_{\min}$. These results confirm that ISAC constellation shaping must be tailored to the specific processing chain of the sensing receiver to achieve the optimal ISAC performance trade-off.

\section{Over-the-Air Experimental Validation} \label{Sec::5}
\begin{figure}[t!]
    \centering
    {\includegraphics[width=0.48\textwidth]{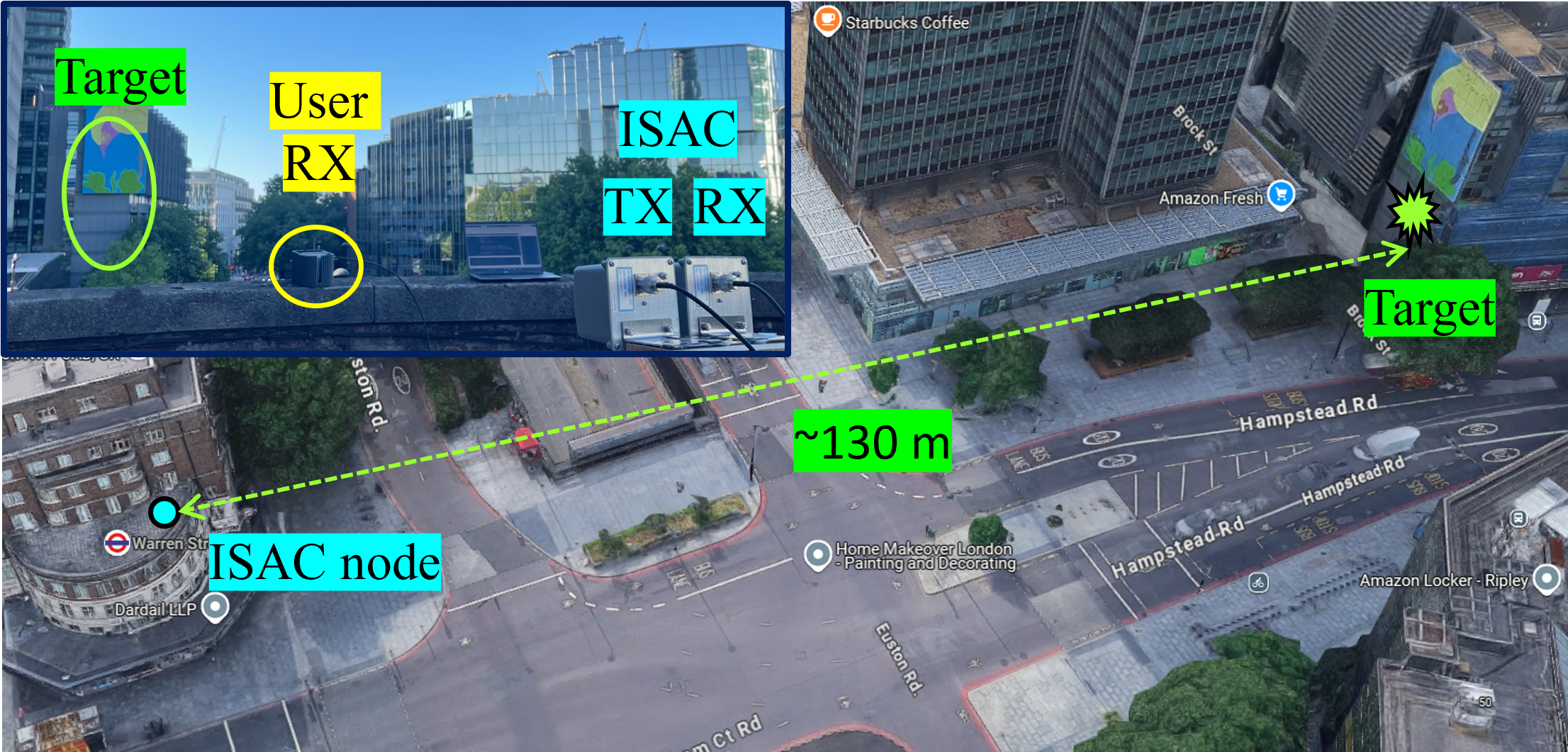}}
    \caption{Photograph of the measurement setup.}
    \label{Fig::5}
\end{figure}

\begin{table}[t!]
    \centering
    \fontsize{8.3}{12}\selectfont
    \caption{OFDM-ISAC Prototype System Parameters.} 
        \label{table4}
    \begin{tabular}{>{\centering\arraybackslash} m{15em} | >{\centering\arraybackslash} m{10em}}
        \toprule
             Specification & Values  \\ 
            \hline
        \midrule
        Center frequency, $f_c$   & 2.4 GHz  \\ \hline
        Bandwidth, $B$   & 20 MHz  \\ \hline
        Number of subcarriers, $N$ & 512 \\ \hline
        CP length & 1.6 $\mu$s \\ \hline
        Symbol duration & 14.4 $\mu$s\\ \hline
        Number of symbols & 512 \\ \hline
        Sampling rate  & 60 MHz \\ \hline
        Antenna gain  & 12 dBi \\ 
        \bottomrule
    \end{tabular}
\end{table}

To bridge the gap between theoretical analysis and practical implementation, we validate the proposed framework through over-the-air experiments using a hardware-in-the-loop prototype. The experimental setup utilizes an AD9363-based software-defined radio (SDR) interfaced with directional antennas providing a 12~dBi gain. The target of interest is a building facade with high electromagnetic reflectivity located approximately 130~m from the ISAC node. As illustrated in Fig. \ref{Fig::5}, the transmitted OFDM-ISAC signal is intercepted by a user receiver while the reflected echoes are processed by the sensing receiver. The system operates at a carrier frequency of 2.4~GHz with a bandwidth of $B = 20$~MHz and $N = 512$ subcarriers. Each transmission frame consists of 512 OFDM symbols, yielding a coherent processing gain of 27~dB. The detailed parameters of the OFDM-ISAC prototype are summarized in Table \ref{table4}. All experimental results are averaged over 100 independent frames to ensure statistical significance.

\subsection{Ranging Performance Evaluation}

\begin{figure}[t!]
    \centering
    \subfigure[]{\includegraphics[width=0.15\textwidth]{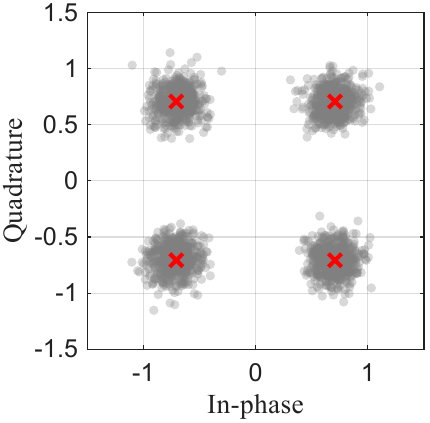}}
    \subfigure[]{\includegraphics[width=0.15\textwidth]{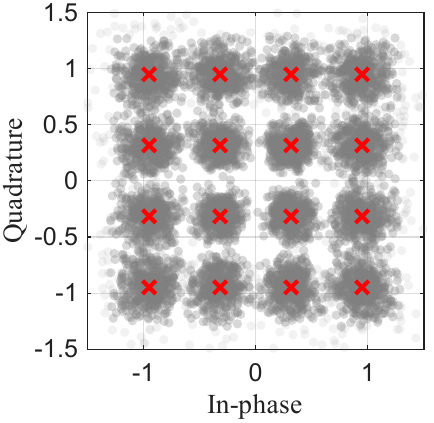}}
    \subfigure[]{\includegraphics[width=0.15\textwidth]{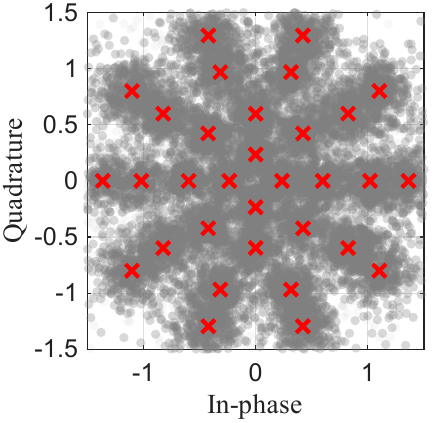}} 
    \caption{Experimental results of OFDM-ISAC: received symbols at the communication user for (a) QPSK, (b) 16QAM, and (C) customized 32APSK.}
    \label{Fig::6}
\end{figure}

\begin{figure}[t!]
    \centering
    {\includegraphics[width=0.48\textwidth]{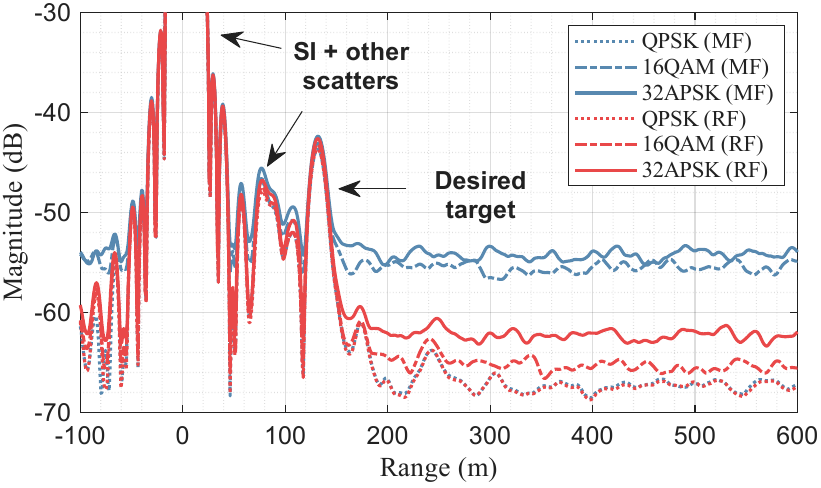}}
    \caption{Measured range profiles obtained by MF and RF processing at the sensing receiver with QPSK, 16QAM, and customized 32APSK.}
    \label{Fig::7}
\end{figure}

The constellation diagrams captured at the user receiver for QPSK, 16QAM, and 32APSK are shown in Fig. \ref{Fig::6}, reflecting the channel-induced distortions and noise. Fig. \ref{Fig::7} presents the corresponding range profiles obtained via MF and RF processing. A key observation is that the RF receiver exhibits a lower effective noise floor compared to the MF receiver, except in the specific case of QPSK modulation. In the MF case, the strong direct-path self-interference (SI) between the transmitter and receiver, coupled with unintended environmental clutter, generates significant sidelobes due to randomly modulated communication data. These sidelobes raise the effective noise floor and can potentially mask the target response. 

In contrast, the RF receiver suppresses these data-dependent sidelobes, although its noise floor is observed to rise in proportion to the inverse second-order moment $\nu_{-2}$ of the constellation. This highlights a critical practical advantage: in environments characterized by dense scatterers and clutter, the RF receiver can provide superior dynamic range when the communication payload employs higher-order constellations ($\mu_4 > 1, \nu_{-2} > 1$), as it remains immune to the cumulative sidelobe interference that degrades the performance of the MF receiver.

\begin{figure}[t!]
    \centering
    {\includegraphics[width=0.48\textwidth]{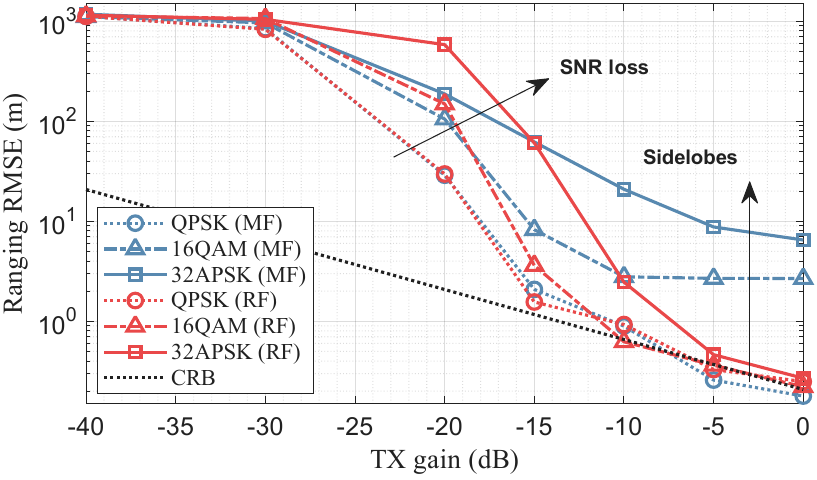}}
    \caption{Measured range estimation accuracy under three different modulation constellations. The reference target range is 130.88~m, estimated with TX gain of 0~dB using QPSK. The CRB is plotted based on the estimated target SNR after post-processing.}
    \label{Fig::8}
\end{figure}

Furthermore, the ranging RMSE for each constellation is evaluated by varying the transmitter gain to control the effective target SNR. As shown in Fig. \ref{Fig::8}, the experimental measurements are in close agreement with both the theoretical analysis and numerical simulations. For the MF receiver, the RMSE reaches an interference-limited saturation point as the SNR increases, whereas the RF receiver maintains a noise-limited trajectory. These experimental results confirm that MF and RF architectures are fundamentally influenced by the constellation geometry and the presence of environmental clutter in distinct ways, necessitating the receiver-specific design approach proposed in this work.

\subsection{Receiver-Specific S\&C Performance Trade-off under GCS}
\begin{figure}[t!]
    \centering
    {\includegraphics[width=0.48\textwidth]{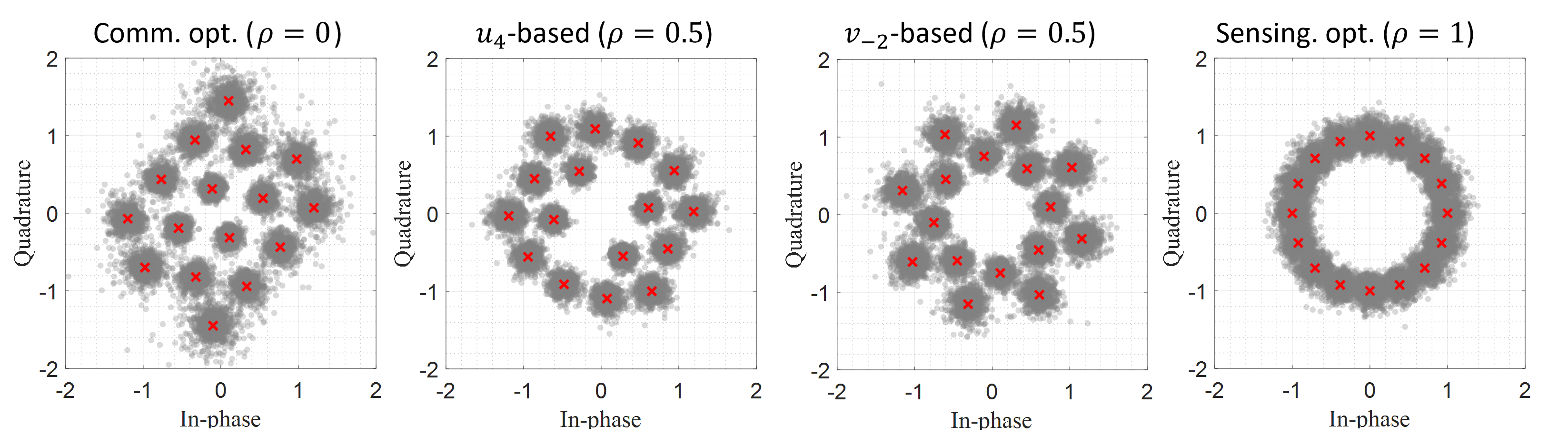}}
    \caption{Measured designed 16GCS constellations at the communication user.}
    \label{Fig::9}
\end{figure}

Finally, we evaluate the proposed ISAC GCS framework through the first reported over-the-air experimental validation of receiver-specific S\&C performance trade-offs. Fig. \ref{Fig::9} presents the captured 16GCS constellations at the communication user for varying priority weights $\rho$. As observed, the communication-optimal design ($\rho=0$) resembles a tilted 16QAM geometry to maximize the MED. Conversely, the sensing-optimal design ($\rho=1$) converges to a 16PSK to minimize magnitude fluctuations. In the joint-optimization regime ($\rho=0.5$), the constellation geometries diverge based on the receiver architecture: the $\mu_4$-optimized constellation for MF maintains a more uniform magnitude distribution, whereas the $\nu_{-2}$-optimized constellation for RF keeps symbols away from the origin to prevent noise amplification.

\begin{figure}[t!]
    \centering
    \subfigure[]{\includegraphics[width=0.45\textwidth]{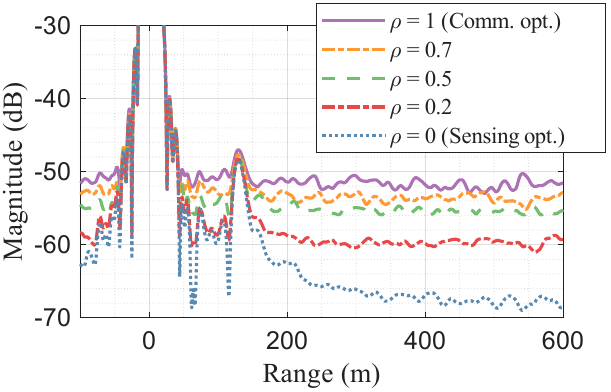}}\\
    \subfigure[]{\includegraphics[width=0.45\textwidth]{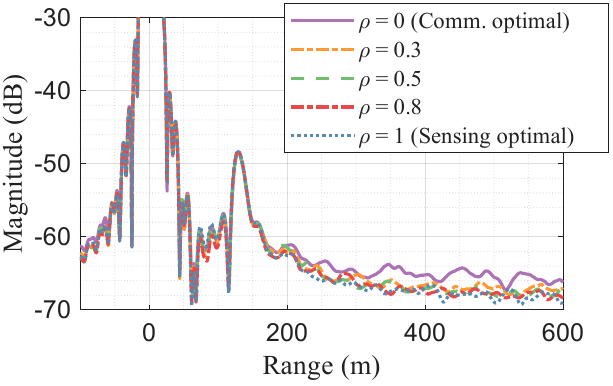}}
    \caption{Measured range profiles with 16GCS under various S\&C priority weights based on (a) MF-based GCS with MF processing and (b) RF-based GCS with RF processing.}
    \label{Fig::10}
\end{figure}

The impact of these designs on the sensing range profile is validated in Fig. \ref{Fig::10}. For the MF-based GCS processed with an MF receiver in Fig. \ref{Fig::10}(a), shifting the priority from communication-optimal ($\rho=0$) to sensing-optimal ($\rho=1$) results in a substantial reduction of the sidelobe level, lowering it from approximately $-50$~dB to nearly $-70$~dB. This confirms that the proposed optimization suppresses random sidelobes by nearly $20$~dB in practical environments, significantly improving the dynamic range for target detection and estimation accuracy. Similarly, for RF-based processing in Fig. \ref{Fig::10}(b), the sensing-optimal design effectively manages the background noise levels. By minimizing the noise enhancement factor $\nu_{-2}$, the GCS allows the RF receiver to maintain a cleaner range profile even as the communication priority increases, preventing the target echo from being buried in amplified thermal noise.

\begin{figure}[t!]
    \centering
    {\includegraphics[width=0.45\textwidth]{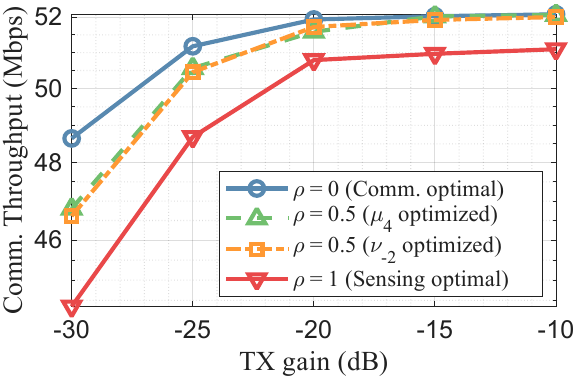}}
    \caption{Measured communication throughput for designed 16GCS under varying TX gain.}
    \label{Fig::11}
\end{figure}

To quantify the communication performance, we evaluate the achieved throughput. The measured communication throughput is calculated as follows:
\begin{equation}
    \eta \; \text{(bits/s)} = \log_2 M \cdot (1 - \text{BER}) \cdot \frac{N}{T_{\text{sym}} + T_{\text{cp}}},
\end{equation}
where $M$ is the constellation size, $N$ is the number of subcarriers, and $T_{\text{sym}}$ and $T_{\text{cp}}$ denote the OFDM symbol and CP durations, respectively.

The measured throughput is characterized in Fig. \ref{Fig::11} as a function of transmitter gain. While the sensing-optimal design ($\rho=1$) incurs a noticeable throughput loss due to the reduced MED inherent in PSK-like structures, the joint designs ($\rho=0.5$) for both MF and RF architectures achieve a throughput remarkably close to the communication-optimal baseline. At high TX gains, these joint designs sustain a throughput of approximately $52$~Mbps, demonstrating that significant sensing precision gains can be realized with negligible impact on data rates.

\begin{figure}[t!]
    \centering
    \subfigure[]{\includegraphics[width=0.45\textwidth]{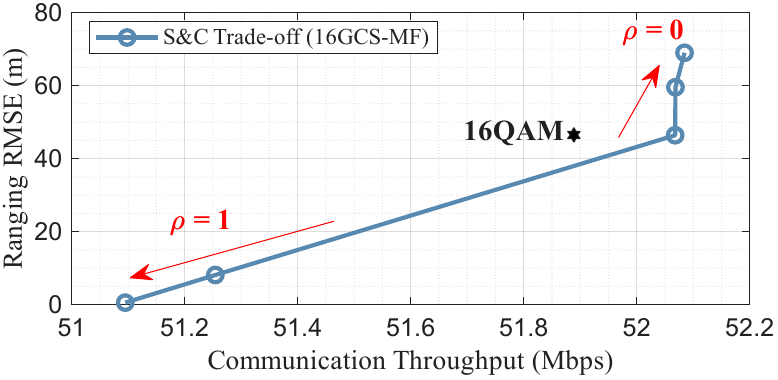}}\\
    \subfigure[]{\includegraphics[width=0.45\textwidth]{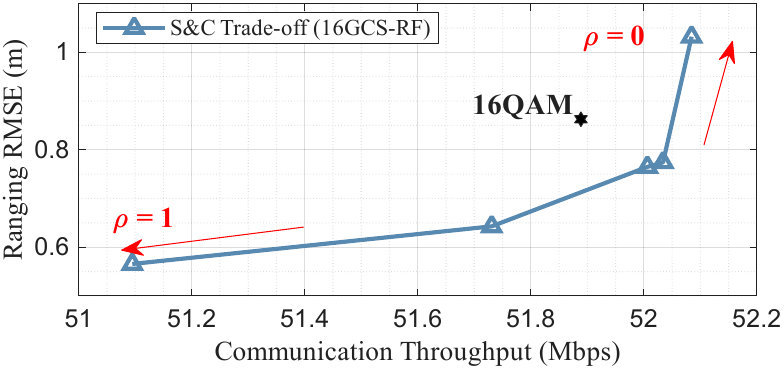}}
    \caption{Measured bottom-line S\&C performance trade-off with the proposed ISAC GCS framework: Ranging RMSE vs. communication throughput. (a) MF-based GCS and (b) RF-based GCS.}
    \label{Fig::12}
\end{figure}

The bottom-line S\&C trade-off is summarized in Fig. \ref{Fig::12}, which depicts ranging RMSE versus communication throughput. In Fig. \ref{Fig::12}(a), the MF-based GCS shows a dramatic RMSE reduction from over $60$~m (communication-optimal) to smaller than $1$~m (sensing-optimal). This highlights the transition from an interference-limited regime to a sidelobe-free regime. In Fig. \ref{Fig::12}(b), the RF-based GCS maintains consistently high precision, with RMSE values staying below $1.1$~m and reaching sub-meter accuracy ($0.6$~m) as $\rho \to 1$. Most importantly, for both receivers, the GCS-optimized curves reside consistently to the lower-right of the standard 16QAM benchmark (black star), signifying higher throughput for a given sensing accuracy. These results represent the first experimental confirmation that receiver-specific OFDM-ISAC constellation shaping successfully navigates the fundamental S\&C trade-off, outperforming standard modulation formats by tailoring the modulation constellation to the specific signal processing chain of the sensing receiver.

\section{Conclusion} \label{Sec::6}
This paper presents a comprehensive analytical and experimental study of data-payload based OFDM-ISAC systems, deriving closed-form range estimation MSE expressions that reveal the distinct impacts of constellation geometry on MF and RF receiver architectures. Our analysis identifies that the MF receiver's sensing performance is fundamentally limited by the constellation kurtosis $\mu_4$ due to data-dependent sidelobes, whereas the RF receiver is governed by the inverse second-order moment $\nu_{-2}$ due to noise amplification. Leveraging these insights, we propose a receiver-specific ISAC GCS framework that optimizes the trade-off between communication throughput and sensing precision. The first reported over-the-air experimental validation of this framework confirms that our GCS-optimized designs successfully extend the S\&C performance trade-off, providing a 20~dB reduction in the sidelobe noise floor for MF processing in practical environment, while achieving sub-meter ranging accuracy. This work provides a practical foundation for designing hardware-compatible and receiver-optimized waveforms in 6G OFDM-ISAC systems.

\appendices
\section{Proof of Theorem \ref{theo1}}\label{proof_theo1}
Let $W = \sum_{n=0}^{N-1} n^2 |x_n|^2$ denote the weighted sum of independent random variables from \eqref{CRB}. Since $W$ is a random variable, we approximate the expectation of its reciprocal using a second-order Taylor expansion about its mean $\mathbb{E}[W]$ as
\begin{align}
    \mathbb{E}\left[ \frac{1}{W} \right] &\approx \frac{1}{\mathbb{E}[W]} + \frac{\mathrm{Var}(W)}{(\mathbb{E}[W])^3} \nonumber \\
    & = \frac{1}{\mathbb{E}[|x_n|^2]\sum_{n=0}^{N-1} n^2} + \frac{\mathrm{Var}(|x_n|^2) \sum_{n=0}^{N-1} n^4}{(\mathbb{E}[|x_n|^2]\sum_{n=0}^{N-1} n^2)^3}, \label{Theo_eq1}
\end{align}
This approximation is valid under the assumption that $W$ is strictly positive and its distribution is sufficiently concentrated around its mean, which is consistent with OFDM signals featuring a large number of subcarriers where the aggregate power $W$ exhibits low relative variance. 

Given that $\mathbb{E}[|x_n|^2] = 1$ and $\mathrm{Var}(|x_n|^2) = \mu_4 - 1$, and utilizing the power sum formulas $\sum_{n=0}^{N-1} n^2 = \frac{(N-1)N(2N-1)}{6}$ and $\sum_{n=0}^{N-1} n^4 = \frac{(N-1)N(2N-1)(3N^2-3N-1)}{30}$, \eqref{Theo_eq1} can be expanded as
\begin{align}
    \mathbb{E}\left[ \frac{1}{W} \right] \approx \frac{6}{N(N-1)(2N-1)} + \frac{36(\mu_4 - 1)(3N^2 - 3N - 1)}{30 \left[ \frac{N(N-1)(2N-1)}{6} \right]^3}.
\end{align}
By applying a first-order asymptotic approximation for large $N$, the expression simplifies to
\begin{align}
    \mathbb{E}\left[ \frac{1}{W} \right] \approx \frac{3}{N^3} + \frac{27(\mu_4 - 1)}{5N^4}.
\end{align}
Substituting this into the CRB expression in \eqref{CRB} yields \eqref{Theo_eq0}.

\section{Proof of Theorem \ref{theo2}}\label{proof_theo2}
Let $C_1$ and $C_2$ denote the sidelobe interference and the noise terms from \eqref{firstDev1}, respectively. To evaluate the MSE, we compute the expectation of the squared magnitude of the first derivative $\mathbb{E}[|\dot{s}_{\text{MF}}(\tau_k)|^2]$. Since the cross-product terms between $C_1$ and $C_2$ are zero-mean and uncorrelated, the total power is the sum of the individual expectations. The expectation of the squared magnitude of $C_1$ is given by
\begin{align}\label{SLMSE}
    \mathbb{E}\left[|C_{1}|^2\right] &= \sum_{j \ne k}^K |\alpha_j|^2 \Bigg( \sum_{n=0}^{N-1} n^2 \mathbb{E}[|x_n|^4] \nonumber \\
    & \quad + \sum_{n = 0}^{N-1} \sum_{m \ne n}^{N-1} n m \mathbb{E}[|x_n|^2] \mathbb{E}[|x_m|^2] e^{j(\phi_{j,n} - \phi_{j,m})} \Bigg) \nonumber \\
    &= \sum_{j \ne k}^K |\alpha_j|^2 \left( \mu_4 \sum_{n=0}^{N-1} n^2 + \left| \sum_{n=0}^{N-1} n e^{j\phi_{j,n}} \right|^2 - \sum_{n=0}^{N-1} n^2 \right) \nonumber \\
    &\approx (\mu_4 - 1) \sum_{j \ne k}^K |\alpha_j|^2 \sum_{n=0}^{N-1} n^2,
\end{align}
where $\phi_{j,n} = 2\pi n \Delta f (\tau_k - \tau_j)$. In \eqref{SLMSE}, we invoke the assumption that for sufficiently resolvable targets ($|\tau_k - \tau_j| \gg 1/B$), the terms $\left| \sum_{n=0}^{N-1} n e^{j\phi_{j,n}} \right|^2$ represent the off-peak values of the range-weighted Dirichlet kernel, which are negligible compared to the sum-of-squares term. 

Next, the expectation of the squared magnitude of the noise term $C_2$ is expressed as
\begin{align}\label{NoiseMSE}
    \mathbb{E}[|C_{2}|^2] &= \sum_{n=0}^{N-1} n^2 \, \mathbb{E}[|x_n^* z_n|^2] = \sum_{n=0}^{N-1} n^2 \, \mathbb{E}[|x_n|^2] \mathbb{E}[|z_n|^2] \nonumber \\
    &= \sigma^2 \sum_{n=0}^{N-1} n^2.
\end{align}
Combining \eqref{SLMSE} and \eqref{NoiseMSE} yields
\begin{align} \label{firstexp}
    \mathbb{E}[|\dot{s}_{\text{MF}}(\tau_k)|^2] = (2\pi \Delta f)^2 \sum_{n=0}^{N-1} n^2 \left( (\mu_4 - 1) \sum_{j \ne k}^K |\alpha_j|^2 + \sigma^2 \right).
\end{align}
Regarding the curvature, the expectation of the second derivative at the true delay $\tau_k$ is obtained as
\begin{align} \label{secondexp}
    \mathbb{E}[\ddot{s}_{\text{MF}}(\tau_k)] = (j2\pi\Delta f)^{2} \alpha_k \sum_{n=0}^{N-1} n^2.
\end{align}
This follows because the interference terms from other scatterers ($j \neq k$) and the noise term vanish in expectation due to the phase de-correlation of the complex exponentials and the zero-mean property of the noise, respectively. Finally, substituting \eqref{firstexp} and \eqref{secondexp} into the general MSE expression in \eqref{MSE3}, we obtain
\begin{align} 
    \mathbb{E}[(\hat{\tau}_k - \tau_k)^2] &\approx \frac{(2\pi\Delta f)^2 \sum_{n=0}^{N-1}n^2 \left( (\mu_{4}-1) \sum_{j\neq k}^K |\alpha_{j}|^2 + \sigma^2 \right)} {2|\alpha_{k}|^2 (2\pi\Delta f)^4 \left(\sum_{n=0}^{N-1}n^2\right)^2} \nonumber \\
    &= \frac{(\mu_{4}-1) \sum_{j\neq k}^K |\alpha_{j}|^2 + \sigma^2} {2|\alpha_{k}|^2 (2\pi\Delta f)^2 \sum_{n=0}^{N-1}n^2}.
\end{align}
Applying the large-$N$ approximation results in
\begin{align}
    \mathrm{MSE}_{\text{MF},k} \approx \frac{3 \left( (\mu_{4}-1) \sum_{j\neq k}^K |\alpha_{j}|^2 + \sigma^2 \right)} {8\pi^2 \Delta f^2 |\alpha_{k}|^2 N^3}.
\end{align}
This completes the proof.

\section{Proof of Theorem \ref{theo3}}\label{proof_theo3}
Following the framework in Lemma \ref{lemma2}, we evaluate the first and second derivatives of the RF correlation output $s_{\text{RF}}(\tau)$ at $\tau_k$. Since the signal term is equalized and independent of random communication data, the first derivative is governed solely by the reshaped noise. Its expected squared magnitude is
\begin{align}
    \mathbb{E}\bigl[\lvert \dot{s}_{\text{RF}}(\tau_k)\rvert^2\bigr] &= (2\pi\Delta f)^2 \sum_{n=0}^{N-1} n^2 \mathrm{Var}(z_{\text{RF},n}) \nonumber \\
    &= (2\pi\Delta f)^2 \sigma^2 \nu_{-2} \sum_{n=0}^{N-1} n^2.
\end{align}
Regarding the curvature, the second derivative is dominated by the target signal component, as the zero-mean noise vanishes in expectation. The square of its expectation is given by
\begin{align}
    \bigl\lvert \mathbb{E}[\ddot{s}_{\text{RF}}(\tau_k)] \bigr\rvert^2 = (2\pi\Delta f)^4 |\alpha_k|^2 \left( \sum_{n=0}^{N-1} n^2 \right)^2.
\end{align}
Substituting these into \eqref{MSE3} yields:
\begin{align}
    \mathrm{MSE}_{\text{RF},k} \approx \frac{\sigma^2 \nu_{-2}}{2 |\alpha_k|^2 (2\pi\Delta f)^2 \sum_{n=0}^{N-1} n^2}.
\end{align}
For large $N$, applying $\sum_{n=0}^{N-1} n^2 \approx N^3/3$ results in the closed-form expression in \eqref{Theo_eq3}, completing the proof.

\bibliographystyle{IEEEtran}
% argument is your BibTeX string definitions and bibliography database(s)
\bibliography{IEEEabrv,reference}
\end{document}